%% file: Main.tex
\documentclass[letterpaper,11pt]{article}
\usepackage[margin=1in]{geometry}
\input{Preamble}

\usepackage[title,toc]{appendix}

\onecolumn

\begin{document}

\title{Polynomial Approximations of Conditional Expectations\\in Scalar Gaussian Channels}
\date{}

\author{Wael Alghamdi and Flavio P. Calmon\thanks{W. Alghamdi and F.P. Calmon are with the John A. Paulson School of Engineering and Applied Sciences at Harvard University. E-mails: \textsf{alghamdi@g.harvard.edu, flavio@seas.harvard.edu}. 
This work was supported in part by the National Science Foundation under Grants CIF 1900750 and CAREER 1845852.}}

\maketitle

\begin{abstract}
    We consider a channel $Y=X+N$ where $X$ is a random variable satisfying $\mathbb{E}[|X|]<\infty$ and $N$ is an independent standard normal random variable. We show that the minimum mean-square error estimator of $X$ from $Y,$ which is given by the conditional expectation $\mathbb{E}[X \mid Y],$ is a polynomial in $Y$ if and only if it is linear or constant; these two cases correspond to $X$ being Gaussian or a constant, respectively. We also prove that the higher-order derivatives of $y \mapsto \mathbb{E}[X \mid Y=y]$ are expressible as multivariate polynomials in the functions $y \mapsto \mathbb{E}\left[ \left( X - \mathbb{E}[X \mid Y] \right)^k \mid Y = y \right]$ for $k\in \mathbb{N}.$ These expressions yield bounds on the $2$-norm of the derivatives of the conditional expectation.  These bounds imply that, if $X$ has a compactly-supported density that is even and decreasing on the positive half-line, then the error in approximating the conditional expectation $\mathbb{E}[X \mid Y]$ by polynomials in $Y$ of degree at most $n$ decays faster than any polynomial in $n.$ 
\end{abstract}

\input{Introduction}

\input{Conditional_Expectation}

\input{Derivatives}

\input{Bernstein}

\clearpage 
\newpage 

\begin{appendices}

\input{Appendix/C_r}

\input{Appendix/Integrability} 

\input{Appendix/Derivative_Bound}

\input{Appendix/Freud}

\input{Appendix/a_n}

\end{appendices}

\bibliographystyle{IEEEtran}
\bibliography{Biblio}

\end{document}

%% file: Preamble.tex
\usepackage[utf8]{inputenc}
\usepackage[english]{babel}

\usepackage{amssymb,amsmath,amsthm,mathrsfs,mathtools}
\usepackage{graphicx}
\usepackage{subfigure}
\usepackage{enumitem}
\usepackage[names]{xcolor}

\usepackage{bm}

\usepackage[square,numbers,sort&compress]{natbib}

\usepackage{hyperref}
\hypersetup{
    colorlinks=true,
    linkcolor=blue,
    filecolor=magenta,      
    urlcolor=cyan,
}

\newtheorem{theorem}{Theorem}

\newtheorem{lemma}{Lemma}
\newtheorem{prop}{Proposition}

\theoremstyle{definition}\newtheorem{definition}{Definition}
\theoremstyle{definition}\newtheorem{remark}{Remark}
\theoremstyle{definition}
\theoremstyle{definition}

\DeclareRobustCommand{\stirling}{\genfrac\{\}{0pt}{}}

\newcommand{\BR}{\mathbb{R}}
\newcommand{\BE}{\mathbb{E}}
\newcommand{\BN}{\mathbb{N}}
\newcommand{\BZ}{\mathbb{Z}}
\newcommand{\SP}{\mathscr{P}}
\newcommand{\SD}{\mathscr{D}}
\newcommand{\SF}{\mathscr{F}}

\newcommand{\supp}{\mathrm{supp}}
\newcommand{\wc}{\mathrel{\;\cdot\;}}

\newcommand{\calC}{\mathcal{C}}
\newcommand{\calD}{\mathcal{D}}

\newcommand{\calN}{\mathcal{N}}

\newcommand{\calP}{\mathcal{P}}

\newcommand{\bg}{\boldsymbol{g}}

\newcommand{\bM}{\boldsymbol{M}}

\newcommand{\bY}{\boldsymbol{Y}}

\newcommand{\bnu}{\pmb{\nu}}
\newcommand{\blambda}{\pmb{\lambda}}

%% file: Introduction.tex
\section{Introduction}

We investigate the extent to which polynomials can approximate conditional expectations in the scalar Gaussian channel. For 
\begin{equation} \label{rk}
    Y = X+N,
\end{equation}
where $X$ has finite variance and $N\sim \calN(0,1)$ is independent of $X$, the conditional expectation $\BE[X \mid Y]$ is the minimum mean-square error (MMSE) estimator: 
\begin{equation}
    \min_{Z} ~ \BE\left[ \left| X - Z \right|^2 \right]  = \BE\left[ \left| X - \BE[X \mid Y] \right|^2 \right],
\end{equation}
where the minimization is taken over all $\sigma(Y)$-measurable random variables $Z.$ It is well-known that $\BE[X \mid Y]$ is linear (i.e., a first degree polynomial in $Y$) if and only if $X$ is Gaussian (see, e.g., \cite{Guo11}). We take this a step further and examine when $\BE[X\mid Y]$ is close to being a polynomial. Specifically, we focus on two questions:
\begin{enumerate}[label=\textbf{(Q\arabic*)}]
    \item For which distributions of $X$ is a polynomial estimator optimal (in the mean-square sense) for reconstructing $X$ from $Y$? \label{sf}
    \item When the MMSE estimator $\BE[X\mid Y]$ is not a polynomial, how well can it be approximated by a polynomial? \label{sd}
\end{enumerate}
In the course of answering \ref{sd}, we answer another fundamental question:
\begin{enumerate}[label=\textbf{(Q3)}]
  \item How can  the higher-order derivatives  of $\BE[X\mid Y=y]$ in $y$ be expressed and bounded? \label{se}
\end{enumerate}

We provide a full answer for \ref{sf} in Theorem~\ref{np}, where we show that the MMSE estimator is a polynomial if and only if $X$ is Gaussian or constant. In other words, the only way $\BE[X \mid Y]$ can be a polynomial is if it is linear in $Y$ or is a constant.

For the second question, if $X$ has a probability density function (PDF) or a probability mass function (PMF) $p_X$ that is compactly-supported, even, and decreasing over $[0,\infty)\cap \supp(p_X),$ then we show in Theorem~\ref{kw} that for all positive integers $n$ and $k$ satisfying $n\ge \max(k-1,1)$ we have that
\begin{equation} \label{rl}
    \inf_{q \in \SP_n} \left\| \BE[X\mid Y] - q(Y) \right\|_2 = O_{X,k}\left( \frac{1}{n^{k/2}} \right).
\end{equation}
Here, $\SP_n$ denotes the set of all polynomials with real coefficients of degree at most $n,$ the implicit constant in \eqref{rl} can depend on both $X$ and $k,$ and $\|\wc\|_2$ refers to the $P_Y$-weighted $2$-norm, i.e., $\|f(Y)\|_2^2=\BE[f(Y)^2].$

The result in \eqref{rl} hinges on our answer to \ref{se} in virtue of it giving a uniform upper bound on the derivatives of the conditional expectation (see Theorem~\ref{og}): there are absolute constants $\{\eta_k\}_{k\ge 1}$ such that
\begin{equation} \label{rm}
    \sup_{\BE[|X|]<\infty} \left\| \frac{d^{k}}{dy^{k}} ~ \BE[X \mid Y=y] \right\|_2 \le \eta_k.
\end{equation}
The bound in \eqref{rm} is a corollary of our answer to the other half of \ref{se}, where we express the derivatives of the conditional expectation in the form (see Proposition~\ref{rs})
\begin{align}
    &\frac{d^{r-1}}{dy^{r-1}} ~ \BE[X \mid Y=y] = \sum_{\substack{2\lambda_2+\cdots+r\lambda_r=r \\ \lambda_2,\cdots,\lambda_r \in \BN}} e_{\lambda_2,\cdots,\lambda_r} \prod_{i=2}^r  \mathbb{E}\left[ \left( X - \mathbb{E}[X \mid Y] \right)^{i} \mid Y = y \right]^{\lambda_i} \label{sc}
\end{align}
for some explicit integers $e_{\lambda_2,\cdots,\lambda_r}$ that we define in the sequel. Setting $r=2$ in \eqref{sc} recovers the first derivative \cite{DytsoDer}
\begin{equation} \label{sh}
    \frac{d}{dy} ~ \BE\left[ X \mid Y = y \right] = \mathrm{Var}\left[ X \mid Y=y \right].
\end{equation}

These results complement our previous work in \cite{Moments}, where we show that if $X$ has a moment generating function (MGF), then there are constants $\{c_{n,j} \}_{n\in \BN, j \in [n]}$ such that
\begin{equation}
    \BE[X\mid Y] = \lim_{n\to \infty} \sum_{j\in [n]} c_{n,j}Y^j
\end{equation}
holds in the mean-square sense. In fact, we may choose
\begin{equation}
    (c_{n,0},\cdots,c_{n,n}) = \BE\left[ (X,XY,\cdots,XY^n) \right] \bM_{Y,n}^{-1}
\end{equation}
where the Hankel matrix of moments of $Y$ is denoted by
\begin{equation}
    \bM_{Y,n}:=\left(\BE \left[ Y^{i+j} \right]\right)_{(i,j)\in [n]^2}.
\end{equation}
Denoting $\bY^{(n)}=(1,Y,\cdots,Y^n)^T,$ the polynomial
\begin{equation}
    E_n[X\mid Y] := \BE\left[ (X,XY,\cdots,XY^n) \right] \bM_{Y,n}^{-1} \bY^{(n)}
\end{equation}
is the orthogonal projection of $\BE[X \mid Y]$ onto the subspace $\SP_n(Y):=\{p(Y) \mid p\in \SP_n\}.$ This projection characterization, in turn, makes $E_n[X \mid Y]$ the best polynomial approximation (in the weighted $L^2$-norm sense) of the conditional expectation $\BE[X \mid Y].$ Specifically, $E_n[X\mid Y]$ uniquely solves the approximation problem
\begin{equation} \label{kp}
    E_n[X\mid Y] = \underset{q(Y)\in \SP_n(Y)}{\mathrm{argmin}} ~  \left\|q(Y) - \BE[X\mid Y] \right\|_2.
\end{equation}

For \eqref{rl}, we apply solutions to the Bernstein approximation problem (see \cite{Lubinsky2007} for a comprehensive survey). The original Bernstein approximation problem extends Weierstrass approximation to polynomial approximation in $L^\infty(\BR,\mu)$ for a measure $\mu$ that is absolutely continuous with respect to the Lebesgue measure. The work by Ditzian and Totik \cite{Lotsch2009}---which introduces moduli of smoothness---shows that tools used to solve the Bernstein approximation problem can also be useful for polynomial approximation in $L^p(\BR,\mu)$ for all $p\ge 1.$ We apply their results for the case $p=2.$

MMSE estimation in Gaussian channels plays a central role in several information-theoretic applications (e.g., \cite{Guo2005,Guo11,Wu2012,Han16,Reeves19}). The MMSE dimension \cite{Wu11} is a measure of nonlinearity of the MMSE estimator. The first-order derivative of the conditional expectation in Gaussian channels has been treated in \cite{DytsoDer}. In particular, formula~\eqref{sh} is generalized in \cite{DytsoDer} to the multivariate case. To the best of our knowledge, no generalization such as~\eqref{sc} to the higher-order derivatives exists in the literature. 

The bound in~\eqref{rl} induces a bound on the gap between the MSE achieved by polynomial estimators and the MMSE. Indeed, the loss from replacing the MMSE estimator $\BE[X\mid Y]$ with its best polynomial approximation $E_n[X\mid Y]$ is
\begin{equation}
    \Delta_{n,X} := \|X-E_n[X\mid Y]\|_2^2 - \|X-\BE[X\mid Y]\|_2^2,
\end{equation}
which satisfies
\begin{equation}
    \Delta_{n,X} \le 2 \|X-E_n[X\mid Y]\|_2 \|E_n[X\mid Y] - \BE[X\mid Y]\|_2.
\end{equation}
Hence,~\eqref{rl} yields the bounds $\Delta_{n,X} = O_{X,\ell}(n^{-\ell})$ for every fixed $\ell>0.$ We note that utilizing higher-order polynomials as proxies of the MMSE has appeared, e.g., in approaches to denoising~\cite{Cha2018}.

Formulas for the conditional expectation that do not require computation of conditional distributions are desirable in practice. For example, the Tweedie formula for the conditional expectation $\BE[X\mid Y=y] = y + p_Y'(y)/p_Y(y)$ helped develop the empirical Bayes method \cite{Robbins}. Similarly, the formula for the higher-order derivatives~\eqref{sc} might shed light on practical applications. For instance, one may obtain a uniform bound $|(d^k/dy^k)\BE[X\mid Y=y]|\le M^kk!$ if, e.g., $X$ is bounded. This implies that the conditional expectation is real analytic. In particular, knowledge of the moments $\BE[X^\ell \mid Y=0]$ (for $\ell \in \BN$) suffices to obtain $\BE[X \mid Y=y]$ on the neighborhood $y\in (-1/M,1/M)$ via Taylor's expansion and the derivative expressions~\eqref{sc}. Further, the value of the conditional expectation $\BE[X\mid Y=y]$ over an interval $y\in (\alpha,\beta)$ is retrievable by its evaluations at only $\lceil M(\beta-\alpha)/2 \rceil+1$ points.

\subsection{Notation} 

The probability measure induced by a random variable (RV) $X$ is denoted by $P_X.$ If $X$ is continuous (resp. discrete), then its PDF (resp. PMF) is denoted by $p_X.$ We use the notation $\|\wc\|_q$ for norms of RVs, i.e., for $q\ge 1$ we have $\left\| X \right\|_q^q = \BE\left[ |X|^q \right].$ We say that a RV $X$ is $n$-times integrable if it satisfies $\|X\|_n<\infty,$ and it is integrable if $\|X\|_1<\infty.$ The norm of the Banach space $L^q(\BR)$ (for $q\ge 1$) is denoted by $\|\wc\|_{L^q(\BR)}.$  

The characteristic function of a RV $Z$ is denoted by $\varphi_Z(t):= \BE\left[ e^{itZ} \right].$ We let $\SP_n$ denote the set of polynomials of degree at most $n$ with real coefficients. For $n\in \BN,$ we set $[n]:=\{0,1,\cdots,n\}$ and denote the set of all finite-length tuples of non-negative integers by $\BN^*$. 

For every integer $r\ge 2,$ let $\Pi_r$ be the set of unordered integer partitions $r=r_1+\cdots+r_k$ of $r$ into integers $r_j\ge 2.$ We encode $\Pi_r$ via a list of the multiplicities of the parts as
\begin{equation} \label{so}
    \Pi_r := \left\{ (\lambda_2,\cdots,\lambda_\ell) \in \BN^{*} ~ ; ~ 2\lambda_2+\cdots+\ell \lambda_\ell = r \right\}.
\end{equation}
In \eqref{so}, $\ell\ge 2$ is free, and trailing zeros are ignored (i.e., $\lambda_{\ell}>0$). For a partition $(\lambda_2,\cdots,\lambda_\ell)=\blambda \in \Pi_r$ having $m=\lambda_2+\cdots+\lambda_\ell$ parts, we denote\footnote{The integer $c_{\blambda}$ counts the number of cyclically-invariant ordered set-partitions of an $r$-element set into $m=\lambda_2+\cdots+\lambda_\ell$ subsets where, for each $k\in \{2,\cdots,\ell\},$ exactly $\lambda_k$ parts have size $k.$}
\begin{equation} \label{ss}
    c_{\blambda} := \frac{1}{m} \binom{m}{\lambda_2,\cdots,\lambda_\ell} \binom{r}{\underbrace{2,\cdots,2}_{\textstyle \lambda_2};\cdots;\underbrace{\ell,\cdots,\ell}_{\textstyle  \lambda_\ell}}
\end{equation}
and
\begin{equation} \label{sp}
    e_{\blambda} := (-1)^{m-1} c_{\blambda}.
\end{equation}
We set\footnote{The integer $C_r$ counts the total number of cyclically-invariant ordered set-partitions of an $r$-element set into subsets of sizes at least $2.$} $C_r:=\sum_{\blambda \in \Pi_r} c_{\blambda}.$ Let $\stirling{r}{m}$ denote the Stirling numbers of the second kind (i.e., the number of unordered set-partitions of an $r$-element set into $m$ nonempty subsets). The integer $C_r$ can be expressed as 
\begin{equation} \label{rx}
    C_r = \sum_{k=1}^r (k-1)!  \sum_{j=0}^k (-1)^{j} \binom{r}{j} \stirling{r-j}{k-j}.
\end{equation}
The first few values of $C_r$ (for $2\le r \le 7$) are given by $1, 1, 4, 11, 56, 267,$ and as $r\to \infty$ we have the asymptotic $C_r \sim (r-1)!/\alpha^r$ for some constant $\alpha \approx 1.146$ (see \cite{OEIS}). The crude bound $C_r<r^r$ can also be seen by a combinatorial argument.  For completeness, equation~\eqref{rx} is derived in Appendix~\ref{aaj}.

\subsection{Assumptions}

We assume only that $X$ is integrable and $N\sim \calN(0,1)$ is independent of $X$ to prove that the conditional expectation $\BE[X\mid X+N]$ cannot be a polynomial of degree exceeding~$1$ (Theorem~\ref{np}) and derive the formula for the higher-order derivatives of the conditional expectation (Proposition~\ref{rs}) along with the ensuing bounds on the norms of the derivatives (Theorem~\ref{og}). For the Bernstein approximation theorem we prove for $\BE[X\mid X+N]$ (Theorem~\ref{kw}), we impose the additional assumption that $X$ is either continuous or discrete with a PDF or a PMF belonging to the set we define next.
\begin{definition} \label{kx}
Let $\SD$ denote the set of compactly-supported even PDFs or PMFs $p$ that are non-increasing over $[0,\infty)~\cap~ \mathrm{supp}(p).$
\end{definition}

%% file: Conditional_Expectation.tex
\section{Polynomial Conditional Expectation} \label{qj}

We start by showing that the only way $\BE[X\mid Y]$ can be a polynomial, for integrable $X$ and $Y=X+N$ a Gaussian perturbation, is if $X$ is Gaussian or constant. The proof is carried in two steps. First, we show that a degree-$m$ non-constant polynomial $\BE[X\mid Y]$ requires $p_Y = e^{-h}$ for some polynomial $h$ with $\deg h = m+1.$ The second step is showing that, because $p_Y=e^{-h}$ is a convolution of the Gaussian kernel, $m =1 .$

The following lemma will be useful for the proof of Theorem~\ref{np}.

\begin{lemma} \label{rv}
For a RV $X$ and a polynomial $p,$ if $p(X)$ is integrable then so is $X^{\deg(p)}.$
\end{lemma}
\begin{proof}
See Appendix~\ref{aae}.
\end{proof}

This lemma will allow us to conclude the finiteness of all moments of $X$ directly from the hypotheses that $\BE[X\mid Y]$ is a polynomial of degree exceeding $1$ and $\|X\|_1<\infty,$ because we have the inequalities $\| \BE[X\mid Y] \|_k \le \|X\|_k$ for every $k\ge 1.$ 

\begin{theorem} \label{np}
For $Y=X+N$ where $X$ is an integrable RV and $N\sim \calN(0,1)$ independent of $X,$ the conditional expectation $\BE[X \mid Y]$ cannot be a  polynomial in $Y$ with degree greater than 1.  Therefore, the MMSE estimator in a Gaussian channel with finite-variance input is a polynomial if and only if the input is Gaussian or constant.
\end{theorem}
\begin{proof}
Suppose, for the sake of contradiction, that
\begin{equation} \label{nk}
    \BE[X \mid Y] = q(Y)
\end{equation}
for some polynomial with real coefficients $q$ of degree $m:=\deg q >1.$ The contradiction we derive will be that the probability measure defined by
\begin{equation} \label{nz}
    Q(B) := \frac{1}{a}\int_B e^{-x^2/2} \, dP_X(x)
\end{equation}
for every Borel subset $B\subset \BR,$ where $a= \BE\left[e^{-X^2/2}\right]$ is the normalization constant, would necessarily have a cumulant generating function that is a polynomial of degree $m+1>2.$ Let $R$ be a RV distributed according to $Q.$ We note that the polynomial $q$ is uniquely determined by \eqref{nk} because $Y$ is continuous, for if $q(Y)=g(Y)$ for a polynomial $g$ then the support of $Y$ must be a subset of the roots of $p-g.$ 

The proof strategy is to compute the PDF $p_Y$ in two ways. One way is to compute $p_Y$  as a convolution
\begin{equation} \label{st}
    p_Y(y) = \frac{1}{\sqrt{2\pi}} \BE\left[ e^{-(X-y)^2/2} \right].
\end{equation}
This equation shows by Lebesgue's dominated convergence that $p_Y$ is continuous. The second way to compute $p_Y$ is via the inverse Fourier transform of $\varphi_Y.$ We consider the Fourier transform that takes an integrable function $\varphi$ to $\widehat{\varphi}(y):=\int_{\mathbb{R}} \varphi(t)e^{-iy t}\, dt,$ so the inverse Fourier transform takes an integrable function $p$ to $(2\pi)^{-1}\int_{\BR} p(y)e^{ity}\, dy.$ Now, $\varphi_Y=\varphi_X\varphi_N$ is integrable; indeed, $|\varphi_X|\le 1$ and $\varphi_N(t)=e^{-t^2/2}.$ Also, being a characteristic function, $\varphi_Y$ is continuous too. Therefore, by the Fourier inversion theorem, since $\varphi_Y/(2\pi)$ is the inverse Fourier transform of $p_Y,$ we obtain that $p_Y = \widehat{\varphi_Y}/(2\pi).$ Equating this latter equation with~\eqref{st}, then multiplying both sides by $\sqrt{2\pi}e^{y^2/2}/a,$ that $R\sim Q$ (see~\eqref{nz}) implies
\begin{equation} \label{sz}
    \BE\left[ e^{Ry} \right] = \frac{1}{a\sqrt{2\pi}}e^{y^2/2} \widehat{\varphi_Y}(y).
\end{equation}
Equation \eqref{sz} holds for every $y\in \BR.$ The rest of the proof derives a contradiction by showing that $\widehat{\varphi_Y}=e^G$ for some polynomial $G$ of degree $m+1>2.$

Integrability of $X$ implies integrability of $\BE[X\mid Y],$ so for every $t\in \BR$
\begin{equation} \label{nj}
    \BE\left[ e^{itY} (X-\BE[X\mid Y]) \right] = 0.
\end{equation}
Substituting $X=Y-N$ and $\BE[X \mid Y] = q(Y)$ into \eqref{nj},
\begin{equation}
    \BE\left[ e^{itY} (Y-N-q(Y)) \right] = 0.
\end{equation}
Because the RVs $e^{itY} (Y-q(Y))$ and $e^{itY}N$ are integrable, we may split the expectation to obtain
\begin{equation} \label{nl}
    \BE\left[ e^{itY}(Y-q(Y)) \right] - \BE\left[e^{itY}N \right] = 0.
\end{equation}
We rewrite equation \eqref{nl} in terms of the characteristic functions of $Y$ and $N.$ 

Since $q(Y)$ is integrable, Lemma~\ref{rv} implies that $Y$ is $m$-times integrable. In particular, we have $\BE\left[ \left| (X+z)^m \right|\right] <\infty$ for some $z\in \BR.$ By Lemma~\ref{rv} again, $X$ is $m$-times integrable. Hence, for each $k\in [m]$ and $Z\in \{X,N,Y\},$ that $\BE\left[ |Z|^k \right]<\infty$ implies that the $k$-th derivative $\varphi_Z^{(k)}$ exists everywhere and
\begin{equation} \label{nm}
    (-i)^k \varphi_Z^{(k)}(t) = \BE\left[ e^{itZ} Z^k \right].
\end{equation}

For the term $\BE\left[ e^{itY}N \right]$ in \eqref{nl}, plugging in $Y=X+N,$ we infer from \eqref{nm} that
\begin{equation} \label{rn}
    \BE\left[ e^{itY}N\right] = \varphi_X(t) \BE\left[ e^{itN}N\right] = -i \varphi_X(t)\varphi_N'(t).
\end{equation}
But $\varphi_N(t) = e^{-t^2/2},$ so $\varphi_N'(t) = -t \varphi_N(t),$ hence \eqref{rn} yields
\begin{equation} \label{nn}
    \BE\left[ e^{itY}N\right] = it \varphi_X(t)\varphi_N(t) = it \varphi_Y(t).
\end{equation}
Let $\alpha_k$ for $k\in [m]$ be real constants such that $q(u) = \sum_{k\in [m]} \alpha_k u^k$ identically over $\BR,$ so $\alpha_m\neq 0.$ For the first term in \eqref{nl}, utilizing \eqref{nm} repeatedly we obtain
\begin{equation} \label{no}
    \BE\left[ e^{itY}(Y-q(Y)) \right] = -i\sum_{k \in [m]} c_k \varphi_Y^{(k)}(t)
\end{equation}
where we define the constants
\begin{equation} \label{nu}
    c_k := (-i)^{k+1}\alpha_k + \delta_{1,k} = \left\{ \begin{array}{cl} 
    (-i)^{k+1}\alpha_k & \text{if } k\in [m]\setminus \{1\}, \\
    1-\alpha_1 & \text{if } k=1.
    \end{array} \right.
\end{equation}
Plugging \eqref{nn} and \eqref{no} in \eqref{nl}, we get the differential equation
\begin{equation} \label{nq}
    t \varphi_Y(t) + \sum_{k \in [m]} c_k \varphi_Y^{(k)}(t) = 0.
\end{equation}
We will transform the differential equation \eqref{nq} into a linear differential equation in the Fourier transform of $\varphi_Y.$ For this, we need first to show that for each $k\in [m]$ the derivative $\varphi_Y^{(k)}$ is integrable so that its Fourier transform is well-defined.

Now, repeated differentiation of $\varphi_Y(t)=\varphi_X(t)e^{-t^2/2}$ shows that for each $k\in [m]$ there is a polynomial $r_k$ in $k+2$ variables such that
\begin{equation} \label{si}
    \varphi_Y^{(k)}(t) = r_k\left( t, \varphi_X(t),\varphi_X'(t), \cdots,\varphi_X^{(k)}(t) \right) e^{-t^2/2}.
\end{equation}
Indeed, we start with $r_0(t,u)=u$ because $\varphi_Y(t)=\varphi_X(t)e^{-t^2/2}.$ Now, suppose \eqref{si} holds for some $k\in [m-1].$ The derivative (with respect to $t$) of the $r_k$ term is
\begin{equation}
    \frac{d}{dt} r_k\left( t, \varphi_X(t), \cdots,\varphi_X^{(k)}(t) \right) = s_k\left( t,\varphi_X(t),\cdots,\varphi_X^{(k+1)}(t) \right)
\end{equation}
for some polynomial $s_k$ in $k+3$ variables. Therefore, differentiating \eqref{si}, we get
\begin{equation}
    \varphi_Y^{(k+1)}(t) = r_{k+1}\left( t, \varphi_X(t),\varphi_X'(t), \cdots,\varphi_X^{(k+1)}(t) \right) e^{-t^2/2}
\end{equation}
where
\begin{align}
    r_{k+1}\left(t,u_0,\cdots,u_{k+1}\right) &:= s_k\left( t,u_0,\cdots,u_{k+1}\right) - t \cdot r_k\left( t,u_0,\cdots,u_k\right)
\end{align}
is a polynomial in $k+3$ variables. Therefore \eqref{si} holds for all $k\in [m].$ Now, for each $j\in [m],$ we have by \eqref{nm} the uniform bound $| \varphi_X^{(j)}(t) | \le \BE\left[ |X|^j \right].$ Therefore, for each $k\in [m],$ letting $v_k$ be the same polynomial as $r_k$ but with the coefficients replaced with their absolute values, the triangle inequality applied to~\eqref{si} yields the bound $| \varphi_Y^{(k)}(t) | \le \eta_k(t) e^{-t^2/2}$ where $\eta_k(t) := v_k\left(|t|,1,\BE[|X|],\cdots,\BE\left[|X|^k\right]\right)$ is a (positive) polynomial in $|t|.$ Since $\int_\BR \eta_k(t)e^{-t^2/2} \, dt <\infty,$ we obtain that $\varphi_Y^{(k)}$ is integrable for each $k\in [m].$

Taking the Fourier transform in the differential equation \eqref{nq} we infer
\begin{equation} \label{nt}
    i\widehat{\varphi_Y}'(y) + \widehat{\varphi_Y}(y) \sum_{k\in [m]} c_k (iy)^k = 0.
\end{equation}
We rewrite this equation in terms of the $\alpha_k$ (see~\eqref{nu}) as
\begin{equation} 
    \widehat{\varphi_Y}'(y) - \widehat{\varphi_Y}(y) \sum_{k\in [m]} (\alpha_k-\delta_{1,k}) y^k = 0.
\end{equation}
Equation \eqref{nt} necessarily implies
\begin{equation}
    \widehat{\varphi_Y}(y) = D \exp\left( \sum_{k\in [m]}  \frac{\alpha_k-\delta_{1,k}}{k+1} y^{k+1}\right)
\end{equation}
for some constant $D.$ Since $p_Y=\widehat{\varphi_Y}/(2\pi),$ we necessarily have $D> 0.$ Therefore, we obtain the desired form for $\widehat{\varphi_Y},$ namely, $\widehat{\varphi_Y}=e^G$ where $G\in \SP_{m+1}\setminus \SP_{m}$ is given by\footnote{It can also be shown that we necessarily have $\alpha_m<0$ and $m$ is odd, but these points are moot since we eventually have a contradiction.}
\begin{equation}
    G(y) := \sum_{k\in [m]} \frac{\alpha_k-\delta_{1,k}}{k+1} y^{k+1} + \log(D).
\end{equation}
Plugging in this formula for $\widehat{\varphi_Y}$ in~\eqref{sz}, we obtain that the cumulant-generating function of the RV $R$ is the degree-$(m+1)$ polynomial $G(y)+y^2/2-\log(a\sqrt{2\pi}),$ contradicting Marcinkiewicz's theorem that a cumulant-generating function has degree at most $2$ if it were a polynomial (see, e.g., \cite[Theorem 2.5.3]{Bryc}). This concludes the proof by contradiction that $\BE[X\mid Y]$ cannot be a polynomial of degree at least $2.$

For the second statement in the theorem, we consider the remaining two cases that $\BE[X\mid Y]$ is a linear expression in $Y$ or is a constant.  If $\BE[X\mid Y]$ is constant, then differentiating and taking the expectation in \eqref{sh} yields that $\|X-\BE[X\mid Y]\|_2=0,$ i.e., $X=\BE[X\mid Y]$ is constant. Finally, under the assumption that $X$ has finite variance, $\BE[X \mid Y]$ is linear if and only if $X$ is Gaussian (see, e.g., \cite{Guo11}). We note that if one requires only that $X$ be integrable, then one may deduce directly from the differential equation \eqref{nq} that a linear $\BE[X\mid Y]$ implies a Gaussian $X$ in this case too, and, for completeness, we end with a proof of this fact.

Assume that $\BE[X\mid Y]=\alpha_1 Y + \alpha_0$ is linear (so $\alpha_1\neq 0$). The differential equation \eqref{nq} becomes
\begin{equation} \label{sj}
    (t-i\alpha_0)\varphi_Y(t) + (1-\alpha_1)\varphi_Y'(t) = 0.
\end{equation}
From \eqref{sj}, we see that $\alpha_1\neq 1,$ because $\varphi_Y$ is nonzero on an open neighborhood around the origin (since $\varphi_Y(0)=1$ and $\varphi_Y$ is continuous). Therefore,
\begin{equation} \label{sk}
    \varphi_Y(t) = C e^{\frac{1}{\alpha_1-1}\left( \frac12 t^2 - i \alpha_0 t \right)},
\end{equation}
for some constant $C.$ Taking $t=0$ in \eqref{sk}, we see that $C=1.$ Therefore, the characteristic function of $Y$ is equal to the characteristic function of a $\calN\left( \frac{\alpha_0}{1-\alpha_1},\frac{1}{1-\alpha_1}\right)$ random variable (by taking $t\to \infty$ in~\eqref{sk}, we get $\alpha_1<1$). In fact, since $\varphi_Y = \varphi_X \cdot \varphi_N,$ we obtain
\begin{equation}
    \varphi_X(t) = e^{-\frac12\cdot \frac{\alpha_1}{1-\alpha_1} \cdot t^2 + it \cdot \frac{\alpha_0}{1-\alpha_1}}.
\end{equation}
Taking $t\to \infty,$ we see that $\alpha_1/(1-\alpha_1)>0,$ i.e., $\alpha_1\in (0,1)$ (note that $\alpha_1\neq 0$ by the assumption that $\BE[X\mid Y]$ is linear).  Therefore, uniqueness of characteristic functions implies that $X$ is Gaussian too.
\end{proof}

%% file: Derivatives.tex
\section{Conditional Expectation Derivatives} \label{pc}

We develop formulas for the higher-order derivatives of the conditional expectation, and establish upper bounds. The bounds in Theorem \ref{og} on the norm of the derivatives of the conditional expectation will be crucial in Section \ref{ji} for establishing a Bernstein approximation theorem that shows how well polynomials can approximate the conditional expectation in the mean-square sense.

\begin{theorem} \label{og}
Fix an integrable RV $X$ and an independent $N\sim \calN(0,1),$ and set $Y=X+N.$ Let $r\ge 2$ be an integer, let $C_r$ be as defined in \eqref{rx}, and denote $q_r:=\lfloor (\sqrt{8r+9}-3)/2 \rfloor$ and $\gamma_r:=(2rq_r)!^{1/(4q_r)}.$ We have the bound
\begin{equation}
    \left\| \frac{d^{r-1}}{dy^{r-1}}~ \BE[X\mid Y=y] \right\|_{2}   \le ~ 2^r C_r \min\left( \gamma_r, \|X\|_{2rq_r}^r \right).
\end{equation}
\end{theorem}

For $2\le r \le 7,$ we obtain the first few values of $q_r$ as $1,1,1,2,2,2,$ and we have $q_r \sim \sqrt{2r}$ as $r\to \infty$ (see Remark~\ref{sy} for a way to reduce $q_r$). To prove Theorem~\ref{og}, we first express the derivatives of $y\mapsto \BE[X\mid Y=y]$ as polynomials in the moments of the RV $X_y-\BE[X_y],$ where $X_y$ denotes the RV obtained from $X$ by conditioning on $\{Y=y\}.$

\begin{prop} \label{rs}
Fix an integrable RV $X$ and an independent $N\sim \calN(0,1),$ and let $Y=X+N.$ For each $(y,k)\in \BR\times \BN,$ denote $f(y) := \BE[X \mid Y=y]$ and
\begin{equation}
    g_k(y) := \mathbb{E}\left[ \left( X - \mathbb{E}[X \mid Y] \right)^k \mid Y = y \right].
\end{equation}
For $(\lambda_2,\cdots,\lambda_\ell)=\blambda\in \BN^*,$ denote $\bg^{\blambda} := \prod_{i=2}^\ell g_i^{\lambda_i},$ with the understanding that $g_i^{0}=1.$ Then, for every integer $r\ge 2,$ we have that
\begin{equation} \label{sw}
    f^{(r-1)} = \sum_{\blambda \in \Pi_r} e_{\blambda} \bg^{\blambda},
\end{equation}
where the integers $e_{\blambda}$ are as defined in~\eqref{ss}-\eqref{sp}.
\end{prop}
\begin{proof}
See Appendix~\ref{aad}.
\end{proof}

Now we are ready to prove Theorem~\ref{og}. 

\begin{proof}[Proof of Theorem~\ref{og}]
We use the notation of Proposition~\ref{rs}. Fix $(\lambda_2,\cdots,\lambda_\ell)=\blambda \in \Pi_r.$ By the generalization of H\"{o}lder's inequality stating $\|\psi_1\cdots \psi_k\|_1 \le \prod_{i=1}^k \|\psi_i\|_k,$ we have that
\begin{equation}  \label{sv}
    \left\| \bg^{\blambda}(Y) \right\|_{2}^{2} = \left\| \prod_{\lambda_i \neq 0} g_i^{2\lambda_i}(Y) \right\|_{1} \le \prod_{\lambda_i\neq 0} \left\| g_i^{2\lambda_i}(Y) \right\|_{s}
\end{equation}
where $s$ is the number of nonzero entries in $\blambda.$ By Jensen's inequality for conditional expectation, for each $i$ such that $\lambda_i\neq 0$ we have that
\begin{equation} \label{su}
    \left\| g_i^{2\lambda_i}(Y) \right\|_{s} \le \| X - \BE[X\mid Y] \|_{2i \lambda_i s}^{2i\lambda_i}.
\end{equation}
Now, $r = \sum_{i=2}^\ell i\lambda_i \ge \sum_{i=2}^{s+1} i = \frac{(s+1)(s+2)}{2}-1,$ so we have that $s^2+3s-2r \le 0,$ i.e., $s\le q_r.$ Further, $i\lambda_i \le r$ for each $i.$ Hence, monotonicity of norms and inequalities \eqref{sv} and \eqref{su} imply the uniform (in $\blambda$) bound
\begin{equation}
    \left\| \bg^{\blambda}(Y) \right\|_{2} \le  \| X - \BE[X\mid Y] \|_{2rq_r}^{r}.
\end{equation}
Observe that $\| X - \BE[X\mid Y] \|_{k}\le 2 \min\left( (k!)^{1/(2k)}, \|X\|_k\right)$ (see \cite{Guo11}), so applying the triangle inequality in~\eqref{sw} we obtain
\begin{align}
    \left\| f^{(r-1)}(Y) \right\|_2 &\le \sum_{\blambda \in \Pi_r} c_{\blambda} \left\| \bg^{\blambda}(Y) \right\|_2 \le 2^rC_r ~ \min\left( \gamma_r, \|X\|_{2rq_r}^r \right),
\end{align}
where $\gamma_r=(2rq_r)!^{1/(4q_r)},$ as desired.
\end{proof}
\begin{remark} \label{sy}
A closer analysis reveals that $i\lambda_i s$ in~\eqref{su} cannot exceed $\beta_r:=t_r^2(t_r+1/2)$ for $t_r:=(\sqrt{6r+7}-1)/3.$ For $r\to \infty,$ we have $rq_r/\beta_r \sim 3^{3/2}/2 \approx 2.6.$ The reduction  when, e.g., $r=7,$ is from $rq_r=14$ to $\beta_r=10.$
\end{remark}

%% file: Bernstein.tex
\section{A Bernstein Approximation Theorem} \label{ji}

We show that, if $p \in \SD$ (see Definition~\ref{kx}) and $X\sim p,$ then the approximation error given by $\|E_n[X\mid Y] - \BE[X\mid Y]\|_2$ decays faster than any polynomial in $n.$

\begin{theorem} \label{kw}
Fix $p \in \SD,$ let $X\sim p,$ suppose $N\sim\calN(0,1)$ is independent of $X,$ and set $Y=X+N.$ There exists a sequence $\{ D(p,k) \}_{k\in \mathbb{N}}$ of constants such that for all integers  $n\ge \max(k-1,1)$ we have
\begin{equation}
    \left\|E_n[X\mid Y] - \mathbb{E}[X \mid Y] \right\|_2 \le \frac{D(p,k)}{n^{k/2}}.
\end{equation}
\end{theorem}

The proof relies on results on the Bernstein approximation problem in weighted $L^p$ spaces. In particular, we consider the Freud case, where the weight is of the form $e^{-Q}$ for $Q$ of polynomial growth, e.g., a Gaussian weight. 

\begin{definition}[Freud Weights] \label{ig}
A function $W:\mathbb{R} \to (0,\infty)$ is called a \emph{Freud Weight}, and we write $W\in \SF,$ if it is of the form $W = e^{-Q}$ for $Q: \mathbb{R}\to \mathbb{R}$ satisfying:
\begin{enumerate}[label=(\arabic*)]
    \item $Q$ is even, \label{aak}
    \item $Q$ is differentiable, and $Q'(y)>0$ for $y>0,$ \label{aal}
    \item $y\mapsto yQ'(y)$ is strictly increasing over $(0,\infty),$ \label{aam}
    \item $yQ'(y)\to 0$ as $y\to 0^+,$ and \label{aan}
    \item \label{aao} there exist $\lambda,a,b,c > 1$ such that for every $y>c$ 
    \begin{equation}
        a \le \frac{Q'(\lambda y)}{Q'(y)} \le b.
    \end{equation} 
\end{enumerate}
\end{definition}

The convolution of a weight in $\SD$ with the Gaussian weight $\varphi(x) := e^{-x^2/2}/\sqrt{2\pi}$ is a Freud weight. This can be shown by noting that with $p_Y=e^{-Q}$ we have $Q'(y)=\BE[N \mid Y=y].$

\begin{theorem} \label{if}
If $p\in \SD$ and $X\sim p,$ then the probability density function of $X+N,$ for $N\sim \mathcal{N}(0,1)$ independent of $X,$ is a Freud weight.
\end{theorem}
\begin{proof}
See Appendix~\ref{aaf}.
\end{proof}

To be able to state the theorem we borrow from the Bernstein approximation literature, we need first to define the Mhaskar–Rakhmanov–Saff number.

\begin{definition} \label{lb}
If $Q:\mathbb{R}\to \mathbb{R}$ satisfies conditions~\ref{aal}--\ref{aan} in Definition~\ref{ig}, and if $yQ'(y)\to \infty$ as $y\to \infty,$ then the \emph{$n$-th Mhaskar–Rakhmanov–Saff number $a_n(Q)$ of $Q$} is defined as the unique positive root $a_n$ of the equation
\begin{equation}
    n = \frac{2}{\pi} \int_0^1 \frac{a_ntQ'(a_nt)}{\sqrt{1-t^2}} \, dt.
\end{equation}
\end{definition}
\begin{remark}
The condition $yQ'(y) \to \infty$ as $y\to \infty$ in Definition \ref{lb} is satisfied if $e^{-Q}$ is a Freud weight. Indeed, in view of properties~\ref{aal}--\ref{aam} in Definition \ref{ig}, the quantity $\ell := \lim_{y\to \infty} yQ'(y)$ is well-defined and it belongs to $(0,\infty].$ If $\ell\neq \infty,$ then because $\lim_{y\to \infty} \lambda y Q'(\lambda y) = \ell$ too, property~\ref{aao} would imply that $a \le 1/\lambda \le b$ contradicting that $\lambda,a>1.$ Therefore, $\ell=\infty.$
\end{remark}

For example, the weight $W(y)=e^{-y^2},$ for which $Q(y)=y^2,$ has $a_n(Q)=\sqrt{n}$ because $\int_0^1 t^2/\sqrt{1-t^2} \, dt = \frac{\pi}{4}.$ If $X\sim p \in \SD,$ say $\supp(p)\subset [-M,M],$ and $p_{Y}=e^{-Q}$ (where $N\sim \calN(0,1)$ is independent of $X,$ and $Y=X+N$), then (see Appendix~\ref{aag})
\begin{equation} \label{sx}
    a_n(Q) \le \left( 2 M +\sqrt{2} \right)\sqrt{n}.
\end{equation}

We apply the following Bernstein approximation theorem \cite[Corollary 3.6]{Lubinsky2007} to prove Theorem~\ref{kw}. 

\begin{theorem} \label{rq}
Fix $W\in \SF,$ and let $u$ be an $r$-times continuously differentiable function such that $u^{(r)}$ is absolutely continuous. Let $a_n=a_n(Q)$ where $W=e^{-Q},$ and fix $1\le s \le \infty.$ Then, for some constant $D(W,r,s)$ and every $n\ge \max(r-1,1)$ 
\begin{equation}
    \inf_{q\in \mathscr{P}_n} \|(q-u)W\|_{L^s(\mathbb{R})} \le D(W,r,s) \left( \frac{a_n}{n} \right)^r \|u^{(r)}W\|_{L^s(\mathbb{R})}.
\end{equation}
\end{theorem}

\begin{proof}[Proof of Theorem~\ref{kw}] 
Fix $k\in \BN$ and $n\ge \max(k-1,1).$ We apply Theorem~\ref{rq} for the function $u(y)=\BE[X\mid Y=y],$ the weight $W=\sqrt{p_Y},$ and for $s=2.$ By our choice of weight, $\|hW\|_{L^2(\BR)} = \|h(Y)\|_2$ for any Borel $h:\BR\to \BR.$  Recall from~\eqref{kp} that $E_n[X\mid Y]$ minimizes $\| q(Y) - \BE[X\mid Y\|_2$ over $q(Y)\in \SP_n(Y).$ By~\eqref{sx}, we have the bound $a_n = O_p(\sqrt{n}).$ Furthermore, by Theorem~\ref{og}, $\|(d^k/dy^k)\BE[X\mid Y]\|_2 = O_k(1).$ Note that $W\in \SF,$ because $W^2=p_Y\in \SF$ by Theorem~\ref{if}. Therefore, by Theorem~\ref{rq}, we obtain a constant $D(p,k)$ such that
\begin{equation}
    \|E_n[X\mid Y] - \mathbb{E}[X \mid Y]\|_2 \le \frac{D(p,k)}{n^{k/2}},
\end{equation}
as desired.
\end{proof}

%% file: Appendix/C_r.tex
\section{A Derivation of Equation \texorpdfstring{\eqref{rx}}{(17)}} \label{aaj}

Using the notation of \cite{partition}, we have that
\begin{equation} \label{aaa}
    C_r = \sum_{k=1}^r (k-1)! \stirling{r}{k}_{\ge 2}
\end{equation}
where $\stirling{r}{k}_{\ge 2}$ denotes the number of partitions of an $r$-element set into $k$ subsets each of which contains at least $2$ elements (note that there are $(k-1)!$ cyclically-invariant arrangements of $k$ parts). The exponential generating function of the sequence $r \mapsto \stirling{r}{k}_{\ge 2}$ is $(e^x-1-x)^k/k!.$ Now, we may write
\begin{equation}
    (e^x-1-x)^k = \sum_{a+b\le k} \binom{k}{a,b} (-1)^{k-a}x^b \sum_{t\in \BN} \frac{(ax)^t}{t!}.
\end{equation}
Therefore, the coefficient of $x^r$ in $(e^x-1-x)^k/k!$ is
\begin{align}
    \frac{1}{r!}\stirling{r}{k}_{\ge 2} &= \sum_{a+b \le k} \frac{(-1)^{k-a}a^{r-b}}{a!b!(k-a-b)!(r-b)!} \\
    &= \frac{1}{r!} \sum_{b=0}^k \binom{r}{b} \sum_{a=0}^{k-b} (-1)^{k-a} \frac{a^{r-b}}{a!(k-a-b)!} \\
    &= \frac{1}{r!} \sum_{b=0}^k \binom{r}{b} \stirling{r-b}{k-b} (-1)^b,
\end{align}
which when combined with~\eqref{aaa} gives~\eqref{rx}.

%% file: Appendix/Integrability.tex
\section{Proof of Lemma~\ref{rv}} \label{aae}

Assume that $\BE\left[|X|^{\deg(p)}\right]=\infty$ (so $\deg(p)\ge 1$), and we will show that $\BE\left[|p(X)|\right] = \infty$ too. Let $k\in [\deg(p)-1]$ be the largest integer for which $\BE\left[ |X|^k \right] < \infty,$ and write $p(u)=u^{k+1}q(u)+r(u)$ for a nonzero polynomial $q$ and a remainder $r\in \SP_k.$ By monotonicity of norms, $\BE\left[ |X|^j \right]<\infty$ for every $j\in [k].$ Hence, $r(X)$ is integrable. Therefore, it suffices to prove that $X^{k+1}q(X)$ is non-integrable, which we show next.

Consider the set $\calD = \left\{ u\in \BR ~ ; ~ \left| q(u) \right| < |a| \right\}$ where $a\neq 0$ is the leading coefficient of $q.$ If $q$ is constant, then $\calD$ is empty, whereas if $\deg q \ge 1$ then $|q(u)|\to \infty$ as $|u|\to \infty$ implies that $\calD$ is bounded; in either case, there is an $M\in \BR$ such that $\calD\subset [-M,M].$ Now, writing $1=1_{\calD}+1_{\calD^c},$ we obtain
\begin{equation} \label{ru}
    \BE\left[ |X|^{k+1} |q(X)| \right] \ge |a| ~ \BE\left[ |X|^{k+1} 1_{\calD^c}(X) \right].
\end{equation}
But we also have that
\begin{equation}
    \infty = \BE\left[ |X|^{k+1}\right] \le M^{k+1} +  \BE\left[ |X|^{k+1} 1_{\calD^c}(X) \right],
\end{equation}
so $\BE\left[ |X|^{k+1} 1_{\calD^c}(X) \right]=\infty.$ Therefore, inequality \eqref{ru} yields that $\BE\left[ \left|X^{k+1}q(X)\right|\right]=\infty,$ concluding the proof.

%% file: Appendix/Derivative_Bound.tex
\section{Proof of Proposition~\ref{rs}} \label{aad}

Recall that the conditional expectation can be expressed as
\begin{equation} \label{aab}
    \BE[Z \mid Y=y] = \frac{\BE\left[ Z e^{-(X-y)^2/2} \right]}{\BE\left[ e^{-(X-y)^2/2} \right]}
\end{equation}
for any RV $Z$ for which $Z e^{-(X-y)^2/2}$ is integrable. This formula applies for both $Z=X$ and $Z=(X-\BE[X\mid Y=y])^k,$ where $(y,k)\in \BR\times \BN,$ because they are polynomials in $X$ and the map $x\mapsto q(x)e^{-(x-y)^2/2}$ is bounded for any polynomial $q.$

Differentiating~\eqref{aab} for $Z=X$ and rearranging terms, we obtain
\begin{equation} \label{ob}
    \frac{d}{dy} ~ \BE[X \mid Y=y] = \frac{\BE\left[ (X-\BE[X \mid Y=y])^2 e^{-(X-y)^2/2} \right]}{\BE\left[ e^{-(X-y)^2/2} \right]},
\end{equation}
i.e., $f'=g_2.$ Note that $g_0 \equiv 1$ and $g_1 \equiv 0.$ Differentiating $g_r$ for $r\ge 1,$ we obtain that
\begin{equation} \label{oc}
    g_r' = g_{r+1} - r g_2 g_{r-1}.
\end{equation} 
We apply successive differentiation to $f'=g_2$ and recover patterns by utilizing~\eqref{oc} at each step. 

From $f'=g_2$ and \eqref{oc}, we infer the first few derivatives
\begin{equation} \label{oe}
    f^{(2)} = g_3, ~ f^{(3)} = g_4-3g_2^2, ~ f^{(4)} = g_5 - 10g_2g_3.
\end{equation}
We see a homogeneity in \eqref{oe}, namely, $f^{(r-1)}$ is an integer linear combination of terms of the form $g_{i_1}^{\alpha_1}\cdots g_{i_\ell}^{\alpha_\ell}$ with $i_1\alpha_1 + \cdots + i_\ell \alpha_\ell = r.$ This homogeneity can be shown to hold for a general $r$ by induction, which we show next. For most of the remainder of the proof, we forget the numerical values of the $f^{(k)}$ and the $g_r^{(k)}$ and only treat them as symbols satisfying $f'=g_2$ and $g_r'=g_{r+1}-rg_2g_{r-1}$ that respect rules of differentiation and which commute.

We call $\sum_{j=1}^\ell i_j \alpha_j$ the \textit{weighted degree} of any nonzero integer multiple of $g_{i_1}^{\alpha_1}\cdots g_{i_\ell}^{\alpha_\ell}.$ This is a well-defined degree because it is invariant to the way the product is arranged. We also say that a sum is of weighted degree $r$ if each summand is of weighted degree $r.$ To prove the claim of homogeneity, i.e., that $f^{(r-1)}$ is of weighted degree $r,$ we differentiate and apply the relation in \eqref{oc} to a generic term $g_{i_1}^{\alpha_1}\cdots g_{i_\ell}^{\alpha_\ell}$ whose weighted degree is $r.$ We have the derivative
\begin{equation} \label{of}
    \left( g_{i_1}^{\alpha_1}\cdots g_{i_\ell}^{\alpha_\ell} \right)' =  \left(g_{i_1}^{\alpha_1}\right)'\cdots g_{i_\ell}^{\alpha_\ell} + \cdots +  g_{i_1}^{\alpha_1}\cdots \left(g_{i_\ell}^{\alpha_\ell} \right)'.
\end{equation}
From \eqref{oc}, for integers $i,\alpha\ge 1,$ 
\begin{equation} \label{oh}
    \left(g_{i}^{\alpha}\right)' = \alpha g_i^{\alpha-1} g_{i+1} - \alpha i g_2 g_{i-1} g_i^{\alpha-1}.
\end{equation}
Therefore, the derivative of $g_i^{\alpha}$ has weighted degree $i\alpha +1.$ In other words, differentiation increased the weighted degree of $g_i^{\alpha}$ by $1.$ From \eqref{of}, then, we see that the weighted degree of $\left( g_{i_1}^{\alpha_1}\cdots g_{i_\ell}^{\alpha_\ell} \right)'$ is $r+1.$ Since $f'=g_2$ is of weighted degree $2,$ induction and linearity of differentiation yield that $f^{(r-1)}$ is of weighted degree $r$ for each $r\ge 2.$

Now, we fix the way we are writing products of the $g_i.$ We ignore explicitly writing $g_0$ and $g_1,$ collect identical terms into an exponent, and write lower indices first. One way to keep this notation is via integer partitions. Consider the ``homogeneous" sets
\begin{equation}
    G_r := \left\{ \sum_{\blambda \in \Pi_r} \beta_{\blambda} \bg^{\blambda} ~ ; ~ \beta_{\blambda} \in \BZ ~ \text{for each} ~ \blambda \in \Pi_r \right\}.
\end{equation}
The homogeneity property for the derivatives of $f$ can be written as $f^{(r-1)} \in G_r$ for each $r\ge 2.$ 

Next, we investigate the exact integer coefficients $h_{\blambda}$ in the expression of the derivatives of $f$ in terms of the $\bg^{\blambda}.$ Homogeneity of the derivatives of $f$ says that we may write each $f^{(r-1)},$ $r\ge 2,$ as an integer linear combination of $\{ \bg^{\blambda}\}_{\blambda \in \Pi_r}.$ One way to obtain such a representation is via repeated differentiation of $f'=g_2,$ applying the relation \eqref{oh}, and discarding any term that is a multiple of $g_1.$ Applying these steps, we arrive at representations
\begin{equation}
    f^{(r-1)} = \sum_{\blambda \in \Pi_r} h_{\blambda} \bg^{\blambda}, ~~~ c_{\blambda} \in \BZ.
\end{equation}

The terms $\bg^{\bnu}$ that appear upon differentiating a term $\bg^{\blambda}$ can be described as follows. For $(\lambda_2,\cdots,\lambda_\ell)=\blambda \in \Pi_r,$ we call $\lambda_2$ the leading term of $\blambda.$ Consider for a tuple $\blambda \in \Pi_r$ the following two sets of tuples $\tau_+(\blambda),\tau_-(\blambda) \subset \Pi_{r+1}$:
\begin{itemize}
    \item The set $\tau_+(\blambda)$ consists of all tuples obtainable from $\blambda$ via replacing a pair $(\lambda_i,\lambda_{i+1})$ with $(\lambda_i-1,\lambda_{i+1}+1)$ (so, necessarily $\lambda_i\ge 1$) while keeping all other entries unchanged;
    
    \item The set $\tau_-(\blambda)$ consists of all tuples obtainable from $\blambda$ via replacing a pair $(\lambda_{i-1},\lambda_i),$ for which $i\ge 3,$ with the pair $(\lambda_{i-1}+1,\lambda_i-1)$ (so, necessarily $\lambda_i \ge 1$) and additionally increasing the leading term by $1$ while keeping all other terms unchanged.
\end{itemize} 
For example, if $\blambda = (0,5,0,1) \in \Pi_{20}$ then
\begin{equation}
    \tau_+(\blambda) = \{ (0,4,1,1),(0,5,0,0,1) \} \subset \Pi_{21}
\end{equation}
and
\begin{equation}
    \tau_-(\blambda) = \{(2,4,0,1),(1,5,1) \} \subset \Pi_{21}.
\end{equation}
The relation \eqref{oh} yields, in view of
\begin{equation} \label{ok}
    \left( g_{2}^{\lambda_2}\cdots g_{\ell}^{\lambda_\ell} \right)' =  \left(g_{2}^{\lambda_2}\right)'\cdots g_{\ell}^{\lambda_\ell} + \cdots +  g_{2}^{\lambda_2}\cdots \left(g_{\ell}^{\lambda_\ell} \right)',
\end{equation}
that
\begin{equation} \label{ol}
    \left( \bg^{\blambda} \right)' = \sum_{\bnu \in \tau_+(\blambda)} a_{\blambda,\bnu}  \bg^{\bnu} - \sum_{\bnu \in \tau_-(\blambda)} b_{\blambda,\bnu} \bg^{\bnu}
\end{equation}
for some positive integers $a_{\blambda,\bnu}$ and $b_{\blambda,\bnu},$ which we describe next. Finding $a_{\blambda,\bnu}$ and $b_{\blambda,\bnu}$ can be straightforwardly done from \eqref{oh} in view of \eqref{ok}. If $\bnu \in \tau_+(\blambda),$ say
\begin{equation}
    (\nu_i,\nu_{i+1}) = (\lambda_i-1,\lambda_{i+1}+1),
\end{equation}
then $a_{\blambda,\bnu} = \lambda_i.$ If $\bnu \in \tau_-(\blambda),$ say
\begin{equation}
    (\nu_{i-1},\nu_{i}) = (\lambda_{i-1}+1,\lambda_{i}-1),
\end{equation}
then $b_{\blambda,\bnu} = i \lambda_i.$ In our example of $\blambda = (0,5,0,1),$ we get
\begin{align}
    a_{(0,5,0,1),(0,4,1,1)} &= 5 \\
    a_{(0,5,0,1),(0,5,0,0,1)} &= 1,
\end{align}
whereas
\begin{align}
    b_{(0,5,0,1),(2,4,0,1)} &= 15 \\
    b_{(0,5,0,1),(1,5,1)} &= 5.
\end{align}
Note that the two sets $\tau_+(\blambda)$ and $\tau_-(\blambda)$ are disjoint because, e.g., the sum of entries of a tuple in $\tau_+(\blambda)$ is the same as that for $\blambda,$ whereas the sum of entries of a tuple in $\tau_-(\blambda)$ is one more than that for $\blambda.$

We next describe how to use what we have shown thus far to deduce a recurrence relation for the  $h_{\blambda}.$ Let $\theta$ be a process inverting $\tau,$ i.e., define for $\bnu \in \Pi_{r+1}$ the two sets
\begin{equation}
    \theta_+(\bnu) := \left\{ \blambda \in \Pi_{r} ~ ; ~ \bnu \in \tau_+(\blambda) \right\}
\end{equation}
and 
\begin{equation}
    \theta_-(\bnu) := \left\{ \blambda \in \Pi_r ~ ; ~ \bnu \in \tau_-(\blambda) \right\}
\end{equation}
The two sets $\theta_+(\bnu)$ and $\theta_-(\bnu)$ are disjoint because the two sets $\tau_+(\blambda)$ and $\tau_-(\blambda)$ are disjoint for each fixed $\blambda.$ Recall our process for defining $h_{\blambda}$: we start with $f'=g_2,$ so $h_{(1)}=1$; we successively differentiate $f'=g_2$; after each differentiation, we use \eqref{oh} and \eqref{ok} (recall that we have the understanding $g_i^0=1$); we discard any ensuing multiple of $g_1$; after $r-2$ differentiations, we get an equation $f^{(r-1)}=\sum_{\blambda \in \Pi_r} h_{\blambda} \bg^{\blambda},$ which we take to be the definition of the $h_{\blambda}.$ The point here is that it could be that $f^{(r-1)}$ is representable as an integer linear combination of the $\bg^{\blambda}$ in more than one way, which can only be verified after the numerical values for the $g_i$ are taken into account, but we are not doing that: our approach treats the $g_i$ as symbols following the laid out rules. Now, we look at one of the steps of this procedure, starting at differentiating $f^{(r-1)} = \sum_{\blambda \in \Pi_r} h_{\blambda} \bg^{\blambda},$ so $f^{(r)} = \sum_{\blambda \in \Pi_r} h_{\blambda} \left( \bg^{\blambda} \right)'.$ Replacing $(\bg^{\blambda})'$ via \eqref{ol},
\begin{equation}
    f^{(r)} = \sum_{\blambda \in \Pi_r} \left( \sum_{\bnu \in \tau_+(\blambda)} a_{\blambda,\bnu}  \bg^{\bnu} - \sum_{\bnu \in \tau_-(\blambda)} b_{\blambda,\bnu} \bg^{\bnu} \right).
\end{equation}
Exchanging the order of summations (for which we use $\theta$),
\begin{equation}
    f^{(r)} = \sum_{\bnu \in \Pi_{r+1}} \left( \sum_{\blambda \in \theta_+(\bnu)} h_{\blambda} a_{\blambda,\bnu} - \sum_{\blambda \in \theta_-(\bnu)} h_{\blambda} b_{\blambda,\bnu} \right) \bg^{\bnu}.
\end{equation}
Therefore, by definition of the $h_{\blambda},$ we have the recurrence: for each $\bnu \in \Pi_{r+1}$ 
\begin{equation} \label{om}
    h_{\bnu} = \sum_{\blambda \in \theta_+(\bnu)} h_{\blambda} a_{\blambda,\bnu} - \sum_{\blambda \in \theta_-(\bnu)} h_{\blambda} b_{\blambda,\bnu}, \quad h_{(1)}=1.
\end{equation}
One instance of this recurrence is, e.g.,
\begin{equation}
    h_{(2,1)} = 3h_{(3)} - 4h_{(1,0,1)}-6h_{(0,2)}.
\end{equation}

Now, we show that the recurrence in~\eqref{om} also generates $e_{\blambda}$ as defined in~\eqref{sp}. For $(\lambda_2,\cdots,\lambda_\ell)=\blambda\in\Pi_r,$ denote $\sigma(\blambda) = \lambda_2+\cdots+\lambda_\ell.$ If $\bnu \in \tau_+(\blambda)$ then $\sigma(\bnu)=\sigma(\blambda),$ and if $\bnu\in \tau_-(\blambda)$ then $\sigma(\bnu)=\sigma(\blambda)+1.$ Therefore, $\blambda \in \theta_+(\bnu)$ implies $\sigma(\bnu)=\sigma(\blambda),$ and $\blambda\in \theta_-(\bnu)$ implies $\sigma(\bnu)=\sigma(\blambda)+1.$ Multiplying~\eqref{om} by $(-1)^{\sigma(\bnu)-1}$ yields the equivalent recurrence
\begin{equation} \label{aac}
    t_{\bnu} = \sum_{\blambda \in \theta_+(\bnu)} t_{\blambda} a_{\blambda,\bnu} + \sum_{\blambda \in \theta_-(\bnu)} t_{\blambda} b_{\blambda,\bnu}, \quad t_{(1)}=1,
\end{equation}
where $t_{\blambda}:= (-1)^{\sigma(\blambda)-1}h_{\blambda}.$ We show that $c_{\blambda}=(-1)^{\sigma(\blambda)-1}e_{\blambda}$ (see~\eqref{sp}) satisfies this recurrence, which is equivalent to $e_{\blambda}$ satisfying the recurrence~\eqref{om}. Clearly, $c_{(1)}=1,$ so consider $c_{\bnu}$ for $\bnu \in \Pi_r$ with $r\ge 3.$

Consider labelled elements $s_1,s_2,\cdots,$ and let $S_k=\{s_1,\cdots,s_k\}$ for each $k\ge 2.$ For any $\blambda\in \Pi_k,$ let $\calC_{\blambda}$ be the set of arrangements of cyclically-invariant set-partitions of $S_k$ according to $\blambda,$ so $|\calC_{\blambda}|=c_{\blambda}.$ Now, fix $\bnu \in \Pi_{r+1},$ and we will build $\calC_{\bnu}$ from the $\calC_{\blambda}$ where $\blambda$ ranges over $\theta_+(\bnu)\cup \theta_-(\bnu).$ Consider first $\blambda\in\theta_+(\bnu),$ where a partition in $\calC_{\bnu}$ is constructed from a partition in $\calC_{\blambda}$ by appending $s_{r+1}$ to one of the parts of the latter partition. Note that adding $s_{r+1}$ to two distinct partitions of $S_r$ cannot produce the same partition of $S_{r+1}$; indeed, just removing $s_{r+1}$ shows that that is impossible. Now, let $i$ be the unique index such that $(\nu_i,\nu_{i+1}) = (\lambda_i-1,\lambda_{i+1}+1).$ Then, a partition $\calP\in \calC_{\bnu}$ of $S_{r+1}$ is induced by a partition $\calP'\in\calC_{\blambda}$ of $S_r$ if and only if $s_{r+1}$ is added to a part in $\calP'$ of size $i,$ of which there are exactly $\lambda_i = a_{\blambda,\bnu}.$ Therefore, we get a contribution of $\sum_{\blambda \in \theta_+(\bnu)} c_{\blambda} a_{\blambda,\bnu}$ towards $c_{\bnu},$ which is the first part in~\eqref{aac}. 

For the second part, $\sum_{\blambda \in \theta_-(\bnu)} c_{\blambda} b_{\blambda,\bnu},$ we consider the remaining ways of generating a partition in $\calC_{\bnu}$ from a partition according to some $\blambda\in \theta_-(\bnu).$ In this case, $s_{r+1}$ is not appended to an existing part, but it is used to create a new part of size $2.$ Thus, we need to also move an element $s_j,$ $1\le j \le r,$ from a part of size at least $3$ to be combined with $s_{r+1}$ to create a new part of size $2.$ It is also clear in this case that such a procedure applied to two distinct partitions in $\calC_{\blambda}$ cannot produce the same partition in $\calC_{\bnu}.$ Let $i$ be the unique index for which $(\nu_{i-1},\nu_{i}) = (\lambda_{i-1}+1,\lambda_{i}-1).$ There are $\lambda_i$ parts to choose from, and $i$ elements to choose from once a part is chosen, so there are a total of $i\lambda_i = b_{\blambda,\bnu}$ ways to generate a partition in $\calC_{\bnu}$ from a partition in $\calC_{\blambda}.$ This gives the second sum in~\eqref{aac}, and we conclude that
\begin{equation}
    c_{\bnu} = \sum_{\blambda \in \theta_+(\bnu)} c_{\blambda} a_{\blambda,\bnu} + \sum_{\blambda \in \theta_-(\bnu)} c_{\blambda} b_{\blambda,\bnu}.
\end{equation}

Therefore, the $c_{\blambda}$ and the $t_{\blambda}$ satisfy the same recurrence, which takes the form: for $\bnu\in \Pi_{r+1}$ there are integers $\{ d_{\blambda,\bnu} \}_{\blambda\in \Pi_r}$ such that
\begin{equation} \label{aah}
    u_{\bnu} = \sum_{\blambda\in \Pi_r} d_{\blambda,\bnu} u_{\blambda}
\end{equation}
with the initial condition $u_{(1)}=1.$ Then, we can induct on $r$ to conclude that the $c_{\blambda}$ and the $t_{\blambda}$ are the same sequence. Since $\Pi_2= \{(1)\},$ we see that $c_{\blambda}=t_{\blambda}$ for every $\blambda \in \Pi_2.$ Suppose $r\ge 2$ is such that $c_{\blambda}=t_{\blambda}$ for every $\blambda \in \Pi_r.$ Hence, for every $\bnu\in \Pi_{r+1},$ we have that 
\begin{equation} \label{aai}
    \sum_{\blambda\in \Pi_r} d_{\blambda,\bnu} c_{\blambda} = \sum_{\blambda\in \Pi_r} d_{\blambda,\bnu} t_{\blambda}.
\end{equation}
Since both sequences $c_{\blambda}$ and $t_{\blambda}$ satisfy the recurrence \eqref{aah}, we obtain from~\eqref{aai} that $c_{\bnu}=t_{\bnu}$ for every $\bnu \in \Pi_{r+1}.$ Therefore, we obtain by induction that $c_{\blambda} = t_{\blambda}$ for every $\blambda\in \Pi_r$ for every $r,$ as desired.

%% file: Appendix/Freud.tex
\section{Proof of Theorem~\ref{if}} \label{aaf}

Fix $p\in \SD,$ suppose $X\sim p,$ and write $Y=X+N$ and $p_Y=e^{-Q}.$ First, we note that $Q'(y)$ is equal to $\mathbb{E}[N \mid Y=y].$

\begin{lemma} \label{ii}
Fix a random variable $X$ and let $Y=X+N$ where $N\sim \mathcal{N}(0,1)$ is independent of $X.$ Writing $p_Y=e^{-Q},$ we have that $ Q'(y) = \mathbb{E}[N \mid Y=y].$
\end{lemma}
\begin{proof}
We have that $p_Y(y) = \mathbb{E}[e^{-(y-X)^2/2}]/\sqrt{2\pi}.$ Differentiating this equation, we obtain that $p_Y'(y) = \mathbb{E}[(X-y)e^{-(y-X)^2/2}]/\sqrt{2\pi},$ where the exchange of differentiation and integration is warranted since $t\mapsto te^{-t^2/2}$ is bounded. 
Now, $Q=-\log p_Y,$ so $Q' = -p_Y'/p_Y,$ i.e.,
\begin{equation} \label{ij}
    Q'(y) = y - \frac{\mathbb{E}[Xe^{-(y-X)^2/2}]}{\mathbb{E}[e^{-(y-X)^2/2}]} = y - \BE\left[ X \mid Y = y \right].
\end{equation}
The proof is completed by substituting $X= Y-N.$
\end{proof}

In view of Lemma~\ref{ii}, that $p$ is even and non-increasing over $[0,\infty)\cap \mathrm{supp}(p)$ imply that $Q$ satisfies conditions \ref{aak}--\ref{aan} of Definition~\ref{ig}. It remains to show that property~\ref{aao} holds. To this end, we show that if $\supp(p) \subset [-M,M]$ and $\lambda=M+2,$ then for every $y>M+4$ we have that  
\begin{equation}
    1<\frac{M^2+5M+8}{2(M+2)} \le \frac{Q'(\lambda y)}{Q'(y)} \le \frac{M^2+7M+8}{4}.
\end{equation}
First, since $Q'(y)=y-\BE[X\mid Y=y]$ (see \eqref{ij}), we have the bounds $y-M \le Q'(y) \le y+M$ for every $y\in \BR.$ Therefore, $y>M$ and $\lambda>1$ imply that
\begin{equation}
    \frac{\lambda y - M}{y+M} \le \frac{Q'(\lambda y)}{Q'(y)} \le \frac{\lambda y +M}{y-M}.
\end{equation}
Further, since $y>M+4$ and $\lambda=M+2,$ we have
\begin{equation}
    \frac{M^2+5M+8}{2(M+2)} < \lambda - \frac{M(M+3)}{y+M} = \frac{\lambda y - M}{y+M}
\end{equation}
and
\begin{equation}
    \frac{\lambda y +M}{y-M} = \lambda + \frac{M(M+3)}{y-M} \le \frac{M^2+7M+8}{4}.
\end{equation}
The fact that $1<\frac{M^2+5M+8}{2(M+2)}$ follows since the discriminant of $M^2+3M+4$ is $-7<0.$ Therefore, $p_Y$ is a Freud weight.

%% file: Appendix/a_n.tex
\section{Proof of Inequality \texorpdfstring{\eqref{sx}}{(52)}} \label{aag}

By Lemma~\ref{ii}, 
\begin{equation}
    Q'(y)= \BE[N\mid Y=y] = y - \BE[X \mid Y=y].
\end{equation}
Therefore $X\le M$ implies that, for any constant $z \ge 0,$ we have
\begin{align}
    \int_0^1 \frac{ztQ'(zt)}{\sqrt{1-t^2}} \, dt &= \frac{\pi}{4} z^2 - z \int_0^1 \hspace{-4pt} \frac{t}{\sqrt{1-t^2}}  \frac{\mathbb{E}\left[ X e^{-(X-zt)^2/2} \right]}{ \mathbb{E}\left[e^{-(X-zt)^2/2}\right]} dt \\
    &\ge \frac{\pi}{4} z^2 - Mz.
\end{align}
We have $\pi z^2/4-Mz>n$ for $z=(2M+\sqrt{2})\sqrt{n}.$ Since $y \mapsto yQ'(y)$ is strictly increasing over $(0,\infty)$ (condition~\ref{aam} of Definition~\ref{ig}), we conclude that $a_n(Q)\le (2M+\sqrt{2})\sqrt{n}.$